\documentclass[a4paper, 11pt]{article}

\usepackage[utf8]{inputenc}
\usepackage[T1]{fontenc}
\usepackage[english]{babel}
\usepackage{authblk}
\usepackage{hyperref}  % autoref
\usepackage{amsfonts, amsthm, amsmath, amssymb} % math packages
\usepackage{lineno}
\usepackage{old-arrows}  % nice arrows
\usepackage{array}  % arrays and tables
\usepackage{subcaption}  % subfigures
\usepackage{color}
\usepackage{pgfplots}
\pgfplotsset{compat=1.14}
\usepackage{stmaryrd}
\usepackage[margin=1.1in, showframe=false]{geometry}  
\usepackage{float}
\usepackage{enumitem}  % list pimping
\usepackage{pdfpages}
\usepackage{tcolorbox}  % problem environment
\usepackage{comment}
\usepackage{graphicx}
\usepackage{graphics}
\usepackage{xspace}
\usepackage[vlined, ruled, dotocloa, english, onelanguage, linesnumbered]{algorithm2e}

\definecolor{asbestos}{RGB}{127, 140, 141}  % light grey
\definecolor{clouds}{RGB}{236, 240, 241}  % very light grey
\definecolor{auroragreen}{RGB}{80, 184, 103}  % green
\definecolor{electron}{RGB}{9, 132, 227}  % blue
\definecolor{matt}{RGB}{140, 122, 230}  % purple

\hypersetup{
	colorlinks=true,
	citecolor=matt,
	linkcolor=asbestos
}

\theoremstyle{plain}
\newtheorem{theorem}{Theorem}  % theorems (bold, sans serif)
\newtheorem*{theorem*}{Theorem}
\newtheorem{proposition}{Proposition}  % propositions
\newtheorem{corollary}{Corollary}  % corollaries 
\newtheorem{definition}{Definition}  % definitions

\theoremstyle{definition}
\theoremstyle{remark}
\newtheorem{example}{Example}  % examples
\tcbuselibrary{many}

\newtcolorbox{problem}[1]{
	colframe=black,
        rightrule=0.6pt, leftrule=0.6pt, toprule=0.6pt, bottomrule=0.6pt,
	titlerule style=\color{black},
	colback=white,
	fonttitle=\color{black},
	arc=0mm,
	enhanced,
	attach boxed title to top left={xshift=0.5cm, yshift=-3.7mm},
	boxed title style={colframe=white},
	title=#1,
	colbacktitle=white,
}

\newtcolorbox{emphasize}{
	colback=clouds,
	colframe=clouds,
	arc=0mm,
	left=0mm,
	right=0mm,
	top=0mm,
	bottom=0mm,
	before skip=5mm,
	after skip=5mm	
}

\newcommand{\Problem}[3]{
	\begin{problem}{#1}
			\begin{tabular}{l p{0.7\textwidth}}
			\textit{Input: \hspace{0.7em}} & #2 \\ %[0.2em]
			\textit{Output: \hspace{0.7em}} & #3 \\	
		\end{tabular}
	\end{problem}
}

\newcommand{\cb}[1]{\mathbb{#1}}  % doubled (IN, IR, ...)
\newcommand{\csf}[1]{\textsf{#1}}  % sans serif
\newcommand{\csmc}[1]{\textsc{#1}}  % small caps
\newcommand{\ctt}[1]{\texttt{\small{#1}}}  % text tt
\newcommand{\eg}{\emph{e.g.}\xspace}  % e.g.
\newcommand{\ie}{\emph{i.e.}\xspace}  % id est
\newcommand{\card}[1]{\vert #1 \vert}  % cardinal, size of a set
\renewcommand{\l}{\ell}  % shortcut for ell
\renewcommand{\max}{{\normalfont \csf{max}}}  % max by inclusion
\newcommand{\G}{G}  % graph
\newcommand{\V}{V}  % vertices
\newcommand{\E}{E}  % edges
\newcommand{\CG}{\csf{CG}}  % conflict-graph

\newcommand{\R}{R}  % relation scheme
\newcommand{\dom}{\csf{dom}}  % domain
\newcommand{\simmodels}{\models_\Phi}
\newcommand{\imp}{\rightarrow}  % implication arrow
\newcommand{\Land}{\bigwedge}  % big and

\newcommand{\NP}{\textbf{\csf{NP}}}  % NP
\title{Functional Dependencies with Predicates: \texorpdfstring{\\}{} What Makes the 
\texorpdfstring{$g_3$-error}{g3-error} Easy to Compute?\footnote{Part of this work was done while the first author was doing a postdoc at LIRIS.}}
\author[1]{Simon Vilmin}
\author[2, 3]{Pierre Faure-{}-Giovagnoli}
\author[2]{Jean-Marc Petit}
\author[2]{Vasile-Marian Scuturici}

\affil[1]{Universit\'e de Lorraine, CNRS, LORIA, F-54000, France}
\affil[2]{Univ Lyon, INSA Lyon, CNRS, UCBL, LIRIS, UMR5205, France}
\affil[3]{Compagnie Nationale du Rh\^one, Lyon, France}

\begin{document}

\maketitle

\begin{abstract}
The notion of functional dependencies (FDs) can be used by data scientists and domain 
experts to confront background knowledge against data.
To overcome the classical, too restrictive, satisfaction of FDs, it is possible to 
replace equality with more 
meaningful binary predicates, and use a coverage measure such as the $g_3$-error to 
estimate the degree to which a FD matches the data.
It is known that the $g_3$-error can be computed in polynomial time if equality is used, 
but unfortunately, the problem becomes \NP-complete when relying on more general 
predicates instead.
However, there has been no analysis of which class of predicates or which properties 
alter the complexity of the problem, especially when going from equality to more general 
predicates.

In this work, we provide such an analysis.
We focus on the properties of commonly used predicates such as equality, similarity 
relations, and partial orders.
These properties are: reflexivity, transitivity, symmetry, and antisymmetry.
We show that symmetry and transitivity together are sufficient to guarantee that the 
$g_3$-error can be computed in polynomial time.
However, dropping either of them makes the problem \NP-complete.

\vspace{0.5em} 

\end{abstract}

\section{Introduction}

Functional dependencies (FDs) are database constraints initially devoted to database 
design \cite{mannila1992design}.
Since then, they have been used for numerous tasks ranging from data 
cleaning \cite{bohannon2007conditional} to data mining \cite{novelli2001functional}.
However, when dealing with real world data, FDs are also a simple yet powerful way to syntactically express background knowledge coming from domain experts 
\cite{faure2022assessing}.
More precisely, a FD $X \imp A$ between a set of attributes (or features) $X$ and another 
attribute $A$ depicts a \emph{function} of the form $f(X) = A$.
In this context, asserting the existence of a function which determines $A$ from $X$ in a 
dataset amounts to testing the validity of $X \imp A$ in a relation, \ie to checking that \emph{every pair} of tuples that are \emph{equal} on $X$ are also \emph{equal} on $A$.
Unfortunately, this semantics of satisfaction suffers from two major drawbacks which 
makes it inadequate to capture the complexity of real world data: (i) it must be checked 
on the whole dataset, and (ii) it uses equality.

%First, the validity of a FD must be verified on the whole dataset.
Drawback (i) does not take into account data quality issues such as outliers, 
mismeasurements or mistakes, which should not impact the relevance of a FD in the data. 
To tackle this problem, it is customary to estimate the partial validity of a given FD 
with a \emph{coverage} measure, rather than its total satisfaction.
The most common of these measures is the $g_3$-error \cite{CormodeCFD,golab2008generating,huhtala1999tane,song2013comparable}, introduced by Kivinen and 
Mannila  \cite{kivinen1995approximate}.
It is the minimum proportion of tuples to remove from a relation in order to satisfy a 
given FD.
As shown for instance by Huhtala et al. \cite{huhtala1999tane}, the $g_3$-error can be 
computed in polynomial time for a single (classical) FD.

%The second drawback of FDs is the central role played by equality.
As for drawback (ii), equality does not always witness efficiently the 
closeness of two real-world values.
%On the one hand, it is likely that every functional dependency can be trivially 
%satisfied on continuous data.
It screens imprecisions and uncertainties that are inherent to every observation.
In order to handle closeness (or difference) in a more appropriate way, numerous 
researches have replaced equality by \emph{binary predicates}, as witnessed by 
recent surveys on relaxed FDs \cite{caruccio2015relaxed,song10tree}.
%This has been done for different extensions of FDs such as 
%differential dependencies \cite{song2010data}, ordered FDs 
%\cite{ng2001extension}, matching dependencies \cite{fan2008dependencies} and more 
%generally relaxed FDs \cite{caruccio2015relaxed, 
%caruccio2016discovery} to mention but a few (see also \cite{song10tree}).

However, if predicates extend FDs in a powerful and meaningful way 
with respect to real-world applications, they also make computations harder.
In fact, contrary to strict equality, computing the $g_3$-error with binary predicates 
becomes \NP-complete \cite{faure2022assessing,song2013comparable}.
In particular, it has been proven for differential 
\cite{song2010data}, matching \cite{fan2008dependencies}, metric \cite{koudas2009metric}, 
neighborhood \cite{bassee2001neighborhood}, and comparable dependencies 
\cite{song2013comparable}.
Still, there is no detailed analysis of what makes the $g_3$-error hard to compute when 
dropping equality for more flexible predicates.
As a consequence, domain experts are left without any insights on which predicates they 
can use in order to estimate the validity of their background knowledge in their data quickly and efficiently.

This last problem constitutes the motivation for our contribution.
In this work, we study the following question: \emph{which properties of 
predicates make the $g_3$-error easy to compute?}
To do so, we introduce binary predicates on each attribute of a relation scheme.
Binary predicates take two values as input and return \ctt{true} or \ctt{false} depending on whether the values match a given comparison criteria.
Predicates are a convenient framework to study the impact of common properties such as 
reflexivity, transitivity, symmetry, and antisymmetry (the properties of equality) on the 
hardness of computing the $g_3$-error.
In this setting, we make the following contributions.
First, we show that dropping reflexivity and antisymmetry does not make the $g_3$-error 
hard to compute.
When removing transitivity, the problem becomes \NP-complete.
This result is intuitive as transitivity plays a crucial role in the computation of the 
$g_3$-error for dependencies based on similarity/distance relations 
\cite{caruccio2015relaxed,song10tree}.
Second, we focus on symmetry.
Symmetry has attracted less attention, despite its importance in partial orders and order 
FDs \cite{dong1982applying,ginsburg1983order,ng2001extension}.
Even though symmetry seems to have less impact than transitivity in the computation of the $g_3$-error, we show that when it is removed the problem also becomes 
\NP-complete.
This result holds in particular for ordered dependencies.
%As a consequence, we deduce that  computing the $g_3$-error for ord
%This result is implicit in numerous works considering the 
%On the negative side, we demonstrate that the problem becomes \NP-complete when removing 
%either transitivity of symmetry.

\textbf{Paper Organization.} In Section \ref{sec:preliminaries}, we recall some 
preliminary definitions.
Section \ref{sec:error} is devoted to the usual $g_3$-error.
In Section \ref{sec:predicates}, we introduce predicates, along with definitions for the
relaxed satisfaction of a functional dependency.
Section \ref{sec:validation} investigates the problem of computing the $g_3$-error when 
equality is replaced by predicates on each attribute.
In Section \ref{sec:discussions} we relate our results with existing extensions of 
FDs.
We conclude in Section \ref{sec:conclusion} with some remarks and open questions for 
further research.
\section{Preliminaries}
\label{sec:preliminaries}

All the objects we consider are finite.
We begin with some definitions on graphs \cite{berge1973graphs} and ordered sets 
\cite{davey2002introduction}. 
A \emph{graph} $\G$ is a pair $(\V, \E)$ where $\V$ is a set of \emph{vertices} and $\E$ is a collection of pairs of vertices called \emph{edges}.
%Unless otherwise stated, $G$ is \emph{undirected}.
An edge of the form $(u, u)$ is called a \emph{loop}.
The graph $\G$ is \emph{directed} if edges are ordered pairs of elements.
Unless otherwise stated, we consider \emph{loopless undirected} graphs.
Let $\G = (\V, \E)$ be an undirected graph, and let $V' \subseteq \V$.
The graph $G[V'] = (V', E')$ with $E' = \{(u, v) \in \E \mid \{u, v\} \subseteq V' \}$ is the graph \emph{induced} by $V'$ with respect to $\G$.
A \emph{path} in $\G$ is a sequence $e_1, \dots, e_m$ of pairwise distinct edges such that $e_i$ and $e_{i+1}$ share a common vertex for each $1 \leq i < m$.
The \emph{length} of a path is its number of edges.
An \emph{independent set} of $\G$ is a subset $I$ of $V$ such that no two vertices in $I$ are connected by an edge of $\G$.
An independent set is \emph{maximal} if it is inclusion-wise maximal among all independent sets.
It is \emph{maximum} if it is an independent set of maximal cardinality.
Dually, a \emph{clique} of $\G$ is a subset $K$ of $\V$ such that every pair of distinct vertices in $K$ are connected by an edge of $\G$.
A graph $\G$ is a \emph{co-graph} if it has no induced subgraph corresponding to a path of length 3 (called $P_4$).
A \emph{partially ordered set} or \emph{poset} is a pair $P = (\V, \leq)$ where $\V$ is a set and $\leq$ a reflexive, transitive, and antisymmetric binary relation.
The relation $\leq$ is called a \emph{partial order}.
If for every $x, y \in \V$, $x \leq y$ or $y \leq x$ holds, $\leq$ is a \emph{total order}.
A poset $P$ is associated to a directed graph $\G(P) = (\V, \E)$ where $(u_i, u_j) \in \E$ exactly when $u_i \neq u_j$ and $u_i \leq u_j$.
An undirected graph $\G = (\V, \E)$ is a \emph{comparability graph} if its edges can be directed so that the resulting directed graph corresponds to a poset.

We move to terminology from database theory 
 \cite{levene2012guided}.
We use capital first letters of the alphabet ($A$, $B$, $C$, ...) to denote attributes and capital last letters (..., $X$, $Y$, $Z$) for attribute sets.
Let $U$ be a universe of attributes, and $\R \subseteq U$ a relation scheme.
Each attribute $A$ in $\R$ takes value in a domain $\dom(A)$.
The domain of $\R$ is $\dom(\R) = \bigcup_{A\in \R} \dom(A)$.
Sometimes, especially in examples, we write a set as a concatenation of its elements (e.g. $AB$ corresponds to $\{A,B\}$).
A \emph{tuple} over $\R$ is a mapping $t \colon \R \to \dom(\R)$ such that $t(A) \in \dom(A)$ for every $A \in R$.
The \emph{projection} of a tuple $t$ on a subset $X$ of $\R$ is the restriction of $t$ to $X$, written $t[X]$.
We write $t[A]$ as a shortcut for $t[\{A\}]$.
A \emph{relation} $r$ over $\R$ is a finite set of tuples over $\R$.
A \emph{functional dependency} (FD) over $\R$ is an expression $X \imp A$ where $X \cup \{A\} \subseteq \R$.
Given a relation $r$ over $\R$, we say that $r$ \emph{satisfies} $X \imp A$, denoted by $r \models X \imp A$, if for every pair of tuples $(t_1, t_2)$ of $r$, $t_1[X] = 
t_2[X]$ implies $t_1[A] = t_2[A]$.
In case when $r$ does not satisfy $X \imp A$, we write $r \not\models X \imp A$.

\section{The \texorpdfstring{$g_3$}{g3}-error}
\label{sec:error}

This section introduces the $g_3$-error, along with its connection with independent sets in graphs through counterexamples and conflict-graphs \cite{bertossi2011database}.

Let $r$ be a relation over $\R$ and $X \imp A$ a functional dependency.
The \emph{$g_3$-error} quantifies the degree to which $X \imp A$ holds in $r$.
We write it as $g_3(r, X \imp A)$.
It was introduced by Kivinen and Mannila \cite{kivinen1995approximate}, and it is 
frequently used to estimate the partial validity of a FD in a dataset \cite{caruccio2015relaxed,CormodeCFD,faure2022assessing,huhtala1999tane}.
It is the minimum proportion of tuples to remove from $r$ to satisfy $X \imp A$, or more formally:

\begin{definition}
Let $\R$ be a relation scheme, $r$ a relation over $\R$ and $X \imp A$ a functional dependency over $\R$.
The \emph{$g_3$-error} of $X \imp A$ with respect to $r$, denoted by $g_3(r, X \imp A)$ is defined as:
    \[ 
        g_3(r, X \imp A) = 1 - \frac{\max(\{ \card{s} \mid s \subseteq r, s \models X \imp A \})}{\card{r}}
    \]
\end{definition}

In particular, if $r \models X \imp A$, we have $g_3(r, X \imp A) = 0$.
We refer to the problem of computing $g_3(r, X \imp A)$ as the \emph{error validation problem} \cite{caruccio2015relaxed,song2013comparable}.
Its decision version reads as follows:
\Problem{Error Validation Problem (EVP)}
{A relation $r$ over a relation scheme $\R$, a FD $X \imp A$, $k \in \cb{R}$.}
{\ctt{yes} if $g_3(r, X \imp A) \leq k$, \ctt{no} otherwise.}

It is known \cite{caruccio2015relaxed,faure2022assessing} that there is a strong relationship between this problem and the 
task of computing the size of a maximum independent set in a graph:
\Problem{Maximum Independent Set (MIS)}
{A graph $\G = (\V, \E)$, $k \in \cb{N}$.}
{\ctt{yes} if $\G$ has a maximal independent set $I$ such that $\card{I} \geq k$, \ctt{no} otherwise.}
To see the relationship between \csmc{EVP} and \csmc{MIS}, we need the notions of \emph{counterexample} and \emph{conflict-graph} \cite{bertossi2011database,faure2022assessing}.
A \emph{counterexample} to $X \imp A$ in $r$ is a pair  of tuples $(t_1, t_2)$ such that $t_1[X] = t_2[X]$ but $t_1[A] \neq t_2[A]$.
%In other words, the pair $(t_1, t_2)$ is a counterexample if $\{t_1, t_2\} \not\models X \imp A$.
The \emph{conflict-graph} of $X \imp A$ with respect to $r$ is the graph $\CG(r, X \imp A) = (r, \E)$ where a (possibly ordered) pair of tuples $(t_1, t_2)$ in $r$ belongs to $\E$ when it is a counterexample to $X \imp A$ in $r$. 
An independent set of $\CG(r, X \imp A)$ is precisely a subrelation of $r$ which satisfies $X \imp A$.
Therefore, computing $g_3(r, X \imp A)$ reduces to finding the size of a maximum independent set in $\CG(r, X \imp A)$.
More precisely, $g_3(r, X \imp A) = 1 - \frac{\card{I}}{\card{r}}$ where $I$ is a maximum independent set of $\CG(r, X \imp A)$.

\begin{example} \label{ex:defs}
    Consider the relation scheme $\R = \{A, B, C, D\}$ with $\dom(\R) = \cb{N}$.
    Let $r$ be the relation over $\R$ on the left of Figure \ref{fig:defs}.
    It satisfies $BC \imp A$ but not $D \imp A$.
    Indeed, $(t_1, t_3)$ is a counterexample to $D \imp A$.
    The conflict-graph $\CG(r, D \imp A)$ is given on the right of Figure \ref{fig:defs}.
    For example,  $\{t_1, t_2, t_6\}$ is a maximum independent set of $\CG(r, D \imp A)$ of maximal size.
    We obtain:
    \[ 
    g_3(r, D \imp A) = 1 - \frac{\card{\{t_1, t_2, t_6\}}}{\card{r}} = 0.5
    \]
    In other words, we must remove half of the tuples of $r$ in order to satisfy $D \imp A$.
    \begin{figure}[ht!]
        \centering 
        \includegraphics[scale=1.0, page=1]{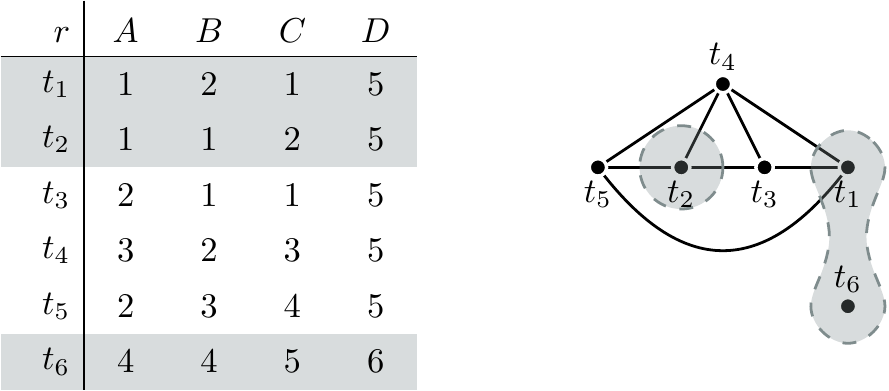}%
        \caption{The relation $r$ and the conflict-graph $\CG(r, D \imp A)$ of Example \ref{ex:defs}.}
        \label{fig:defs}
    \end{figure}
\end{example}

However, \csmc{MIS} is an \NP-complete problem \cite{garey1979computers} while computing $g_3(r, X \imp A)$ takes polynomial time in the size of $r$ and $X \imp A$ \cite{huhtala1999tane}.
This difference is due to the properties of equality, namely reflexivity, transitivity, symmetry and antisymmetry.
They make $\CG(r, X \imp A)$ a disjoint union of complete $k$-partite graphs, and hence a co-graph \cite{faure2022assessing}.
In this class of graphs, solving \csmc{MIS} is polynomial \cite{giakoumakis1997p_4}. 
This observation suggests to study in greater detail the impact of such properties on the structure of conflict-graphs.
%We will study in details the impact of each properties on the structure of 
%conflict-graphs.
First, we need to introduce predicates to relax equality, and to define a more general version of the error validation problem accordingly.
%This is the aim of the next section.
%
\section{Predicates to relax equality}
\label{sec:predicates}

In this section, in line with previous researches on extensions of functional dependencies \cite{song10tree,caruccio2015relaxed}, we equip each attribute of a relation scheme with a binary predicate.
We define the new $g_3$-error and the corresponding error validation problem.
%As mentioned in the introduction, predicates estimate the similarity 
%(or difference) of two values in a way which reflects better the complexity of real 
%world 
%data, as well as the knowledge of domain experts.

Let $\R$ be a relation scheme.
For each $A \in \R$, let $\phi_A \colon \dom(A) \times \dom(A) \to \{\ctt{true}, \ctt{false}\}$ be a predicate.
%When clear from the context, we write $\phi_A(x, y)$ as a shortcut for $\phi_A(x, y) = 
%\ctt{true}$ and $\overline{\phi_A(x, y)}$ for $\phi_A(x, y) = \ctt{false}$.
For instance, the predicate $\phi_A$ can be equality, a distance, or a similarity relation.
We assume that predicates are black-box oracles that can be computed in polynomial time in the size of their input.

Let $\Phi$ be a set of predicates, one for each attribute in $\R$.
The pair $(\R, \Phi)$ is a \emph{relation scheme with predicates}.
In a relation scheme with predicates, relations and FDs are unchanged.
%In other words, a relation remains a collection of tuples over $\R$, and a functional 
%dependency remains an expression of the form $X \imp A$, where $X \cup \{A\} 
%\subseteq \R$.
However, the way a relation satisfies (or not) a FD can easily be adapted 
to $\Phi$.

\begin{definition}[Satisfaction with predicates] \label{def:sat-with-predicates}
    Let $(\R, \Phi)$ be a relation scheme with predicates, $r$ a relation and $X \imp A$ a functional dependency both over $(\R, \Phi)$.
    The relation $r$ \emph{satisfies} $X \imp A$ with respect to $\Phi$, denoted by $r \models_{\Phi} X \imp A$, if for every pair of tuples $(t_1, t_2)$ of $r$, the following formula holds:
    \[ 
        \left(\bigwedge_{B \in X} \phi_B(t_1[B], t_2[B])\right) \implies \phi_A(t_1[A], t_2[A])
    \]
    %where $\bigwedge$ is the classical and operation.
    %
\end{definition}

An new version of the $g_3$-error adapted to $\Phi$ is presented in the following definition.

\begin{definition}
Let $(\R, \Phi)$ be a relation scheme with predicates, $r$ be a relation over $(\R, \Phi)$ and $X \imp A$ a functional dependency over $(\R, \Phi)$.
The \emph{$g_3$-error with predicates} of $X \imp A$ with respect to $r$, denoted by $g_3^{\Phi}(r, X \imp A)$ is defined as:
\[ 
    g_3^{\Phi}(r, X \imp A) = 1 - \frac{\max(\{ \card{s} \mid s \subseteq r, s \simmodels X \imp A \})}{\card{r}}
\]
\end{definition}

From the definition of $g_3^{\Phi}(r, X \imp A)$, we derive the extension of the error validation problem from equality to predicates:
\Problem{Error Validation Problem with Predicates (EVPP)}
{A relation $r$ over a relation scheme with predicates $(\R, \Phi)$, a FD $X \imp A$ over $(\R, \Phi)$, $k \in \cb{R}$.}
{\ctt{yes} if $g_3^{\Phi}(r, X \imp A) \leq k$, \ctt{no} otherwise.} 
Observe that according to the definition of satisfaction with predicates (Definition 
\ref{def:sat-with-predicates}), counterexamples and conflict-graphs remain well-defined.
However, for a given predicate $\phi_A$, $\phi_A(x, y) = \phi_A(y, x)$ needs not be true 
in general, meaning that we have to consider ordered pairs of tuples. 
That is, an ordered pair of tuples $(t_1, t_2)$ in $r$ is a counterexample to $X \imp A$ 
if $\bigwedge_{B \in X} \phi_B(t_1[B], \allowbreak t_2[B]) = \ctt{true}$ but 
$\phi_A(t_1[A], 
\allowbreak t_2[A]) \neq \ctt{true}$.

We call $\CG_{\Phi}(r, X \imp A)$ the conflict-graph of $X \imp A$ in $r$.
In general, $\CG_{\Phi}(r, X \imp A)$ is directed. 
It is undirected if the predicates of $\Phi$ are symmetric (see Section 
\ref{sec:validation}).
In particular, computing $g_3^{\Phi}(r, X \imp A)$ still amounts to finding the size of a 
maximum independent set in $\CG_{\Phi}(r, X \imp A)$.

\begin{example} \label{ex:def-g3}
    We use the relation of Figure \ref{fig:defs}.
    Let $\Phi = \{\phi_A, \phi_B, \phi_C, \phi_D\}$ be the collection of predicates defined as follows, for every $x, y \in \cb{N}$:
    \begin{itemize}
        \item $\phi_A(x, y) = \phi_B(x, y) = \phi_C(x, y) = \ctt{true}$ if and only if $\card{x - y} \leq 1$. 
        Thus, $\phi_A$ is reflexive and symmetric but not transitive (see Section 
\ref{sec:validation}), 
        \item $\phi_D$ is the equality.
    \end{itemize}
    The pair $(R, \Phi)$ is a relation scheme with predicates.
    We have $r \simmodels AB \imp D$ but $r \not\simmodels C \imp A$.
    In Figure \ref{fig:ex-defs-predicates}, we depict $\CG_{\Phi}(r, C \imp A)$.
    A maximum independent set of this graph is $\{t_1, t_2, t_3, t_5\}$.
    We deduce
    \[ g_3^{\Phi}(r, C \imp A) = 1 - \frac{\card{\{t_1, t_2, t_3, 
    t_5\}}}{\card{r}} = \frac{1}{3} \]
    \begin{figure}[ht!]
        \centering 
        \includegraphics[page=2, scale=1.0]{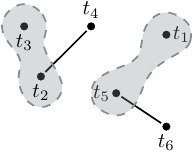}%
        \caption{The conflict-graph $\CG_{\Phi}(r, C \imp A)$ of Example \ref{ex:def-g3}.}
        \label{fig:ex-defs-predicates}
    \end{figure}
\end{example}

Thus, there is also a strong relationship between \csmc{EVPP} and \csmc{MIS}, similar to the one between \csmc{EVP} and \csmc{MIS}.
Nonetheless, unlike \csmc{EVP}, the problem \csmc{EVPP} is \NP-complete \cite{song2013comparable}.
In the next section, we study this gap of complexity between \csmc{EVP} and \csmc{EVPP} via different properties of predicates.

\section{Predicates properties in the \texorpdfstring{$g_3$-error}{g3-error}}
\label{sec:validation}

In this section, we study properties of binary predicates that are commonly used to replace equality.
We show how each of them affects the error validation problem. 
%More precisely, we prove that the task can still be conducted in polynomial time if 
%reflexivity is dropped.
%Then, we prove that removing either symmetry or transitivity makes \csmc{EVPP} an 
%\NP-complete problem.
%Our results are summarized in Figure \ref{fig:summary}.

First, we define the properties of interest in this paper.
Let $(\R, \Phi)$ be a relation scheme with predicates.
Let $A \in \R$ and $\phi_A$ be the corresponding predicate.
We consider the following properties:
\begin{itemize}[leftmargin=3em]
    \item[(\ctt{ref})] $\phi_A(x, x) = \ctt{true}$ for all $x \in \dom(A)$ (reflexivity)
    \item[(\ctt{tra})] for all $x, y, z \in \dom(A)$, $\phi_A(x, y) = \phi_A(y, z) = 
    \ctt{true}$ implies $\phi_A(x, z) = \ctt{true}$ (transitivity)
    \item[(\ctt{sym})] for all $x, y \in \dom(A)$, $\phi_A(x, y) = \phi_A(y, x)$ (symmetry)
    \item[(\ctt{asym})] for all $x, y \in \dom(A)$, $\phi_A(x, y) = \phi_A(y, x) = 
    \ctt{true}$ implies $x = y$ (antisymmetry).
\end{itemize}
Note that symmetry and antisymmetry together imply transitivity as $\phi_A(x, y) = \ctt{true}$ entails $x = y$.

As a first step, we show that symmetry and transitivity are sufficient to make \csmc{EVPP} solvable in polynomial time.
In fact, we prove that the resulting conflict-graph is a co-graph, as with equality.

%We begin our investigation of \csmc{EVPP} by considering reflexivity 
%(\ctt{ref}).
%We show that removing this property does not make \csmc{EVPP} harder than \csmc{EVP}.

\begin{theorem} \label{thm:vpp-ref}
The problem \csmc{EVPP} can be solved in polynomial time if the predicates used on each attribute are transitive (\ctt{tra}) and symmetric (\ctt{sym}).
\end{theorem}

\begin{proof}
Let $(R, \Phi)$ be a relation scheme with predicates.
Let $r$ be relation over $(\R, \Phi)$ and $X \imp A$ be a functional dependency, also over $(\R, \Phi)$.
We assume that each predicate in $\Phi$ is transitive and symmetric.
We show how to compute the size of a maximum independent set of $\CG_{\Phi}(r, X \imp A)$ 
in polynomial time.

As $\phi_A$ is not necessarily reflexive, a tuple $t$ in $r$ 
can produce a counter-example $(t, t)$ to $X \imp A$.
Indeed, it may happen that $\phi_B(t[B], t[B]) = \texttt{true}$ for each $B \in X$, but 
$\phi_A(t[A], t[A]) = \texttt{false}$.
However, it follows that $t$ never belongs to a subrelation $s$ of $r$ satisfying 
$s \simmodels X \imp A$.
Thus, let $r' = r \setminus \{t \in r \mid \{t\} \not\simmodels X \imp A\}$.
Then, a subrelation of $r$ satisfies $X \imp A$ if and only if it is an independent set 
of $\CG_{\Phi}(r, X \imp A)$ if and only if it is an independent set of $\CG_{\Phi}(r', X 
\imp A)$.
Consequently, computing $g_3^{\Phi}(r, X \imp A)$ is solving \csmc{MIS} in 
$\CG_{\Phi}(r', X \imp A)$.

We prove now that $\CG_{\Phi}(r', X \imp A)$ is a co-graph.
Assume for contradiction that $\CG_{\Phi}(r', X \imp A)$ has an induced path $P$ with $4$ 
elements, say $t_1, t_2, t_3, t_4$ with edges $(t_1, t_2)$, $(t_2, t_3)$ and $(t_3, 
t_4)$.
Remind that edges of $\CG_{\Phi}(r', X \imp A)$ are counterexamples to $X \imp A$ 
in $r'$. 
Hence, by symmetry and transitivity of the predicates of $\Phi$, we deduce that for each 
pair $(i, j)$ in $\{1, 2, 3, 4\}$, $\Land_{B \in X} \phi_B(t_i[B], t_j[B]) = \ctt{true}$.
Thus, we have $\Land_{B \in X} \phi_B(t_3[B], t_1[B]) = \Land_{B \in 
X} \allowbreak \phi_B(t_1[B], \allowbreak t_4[B]) = \ctt{true}$.
However, neither $(t_1, t_3)$ nor $(t_1, t_4)$ belong to $\CG_{\Phi}(r', \allowbreak X \imp A)$ since 
$P$ is an induced path by assumption.
Thus, $\phi_A(t_3[A], t_1[A]) = \phi_A(t_1[A], t_4[A]) = \ctt{true}$ must hold.
Nonetheless, the transitivity of $\phi_A$ implies $\phi_A(t_3[A], t_4[A]) = 
\ctt{true}$, a contradiction with $(t_3, t_4)$ being an edge of $\CG_{\Phi}(r', 
\allowbreak X \imp A)$.
We deduce that $\CG_{\Phi}(r', X \imp A)$ cannot contain an induced $P_4$, 
and that it is indeed a co-graph.
As \csmc{MIS} can be solved in polynomial time for co-graphs \cite{giakoumakis1997p_4}, 
the theorem follows.
\end{proof}

One may encounter non-reflexive predicates when dealing with strict 
orders or with binary predicates derived from \ctt{SQL} equality. 
In the $3$-valued logic of \ctt{SQL}, comparing the \ctt{null} value with itself 
evaluates to \ctt{false} rather than \ctt{true}.
With this regard, it could be natural for domain experts to use a predicate which is 
transitive, symmetric and reflexive almost everywhere but on the \ctt{null} value.
This would allow to deal with missing information without altering the data.

The previous proof heavily makes use of transitivity, which has a strong impact on the 
edges belonging to the conflict-graph.
Intuitively, conflict-graphs can become much more complex when transitivity is dropped.
Indeed, we prove an intuitive case: when predicates are not required to be 
transitive, \csmc{EVPP} becomes intractable.
%The proof is derived from the proof of Song \cite{song2010data} for differential 
%dependencies.

\begin{theorem} \label{thm:vpp-tra}
The problem \csmc{EVPP} is \NP-complete even when the predicates used on each 
attribute are symmetric (\ctt{sym}) and reflexive (\ctt{ref}).
\end{theorem}

The proof is given in Appendix \ref{app:vpp-tra}.
It is a reduction from the problem (dual to MIS) of finding
the size of a maximum clique in general graphs. 
It uses arguments similar to the proof of Song et al. \cite{song2013comparable} showing 
the \NP-completeness of \csmc{EVPP} for comparable dependencies.

We turn our attention to the case where symmetry is dropped from the predicates.
In this context, conflict-graphs are directed.
Indeed, an ordered pair of tuples $(t_1, t_2)$ may be a counterexample to a functional 
dependency, but not $(t_2, t_1)$.
Yet, transitivity still contributes to constraining the structure of 
conflict-graphs, as suggested by the following example.
\begin{example} \label{ex:ex-sym}
    We consider the relation of Example \ref{ex:defs}. 
    We equip $A, B, C, D$ with the following predicates:
    \begin{itemize}
        \item $\phi_C(x, y) = \ctt{true}$ if and only if $x \leq y$
        \item $\phi_A(x, y)$ is defined by
        \[ 
        \phi_A(x, y) =
        \begin{cases}
            \ctt{true} & \text{if } x = y \\
            \ctt{true} & \text{if } x = 1 \text{ and } y \in \{2, 4\} \\
            \ctt{true} & \text{if } x = 3 \text{ and } y = 4 \\
            \ctt{false} & \text{otherwise.}
        \end{cases}
        \]
        \item $\phi_B$ and $\phi_D$ are the equality.
    \end{itemize}
    Let $\Phi = \{\phi_A, \phi_B, \phi_C, \phi_D\}$.
    The conflict-graph $\CG_{\Phi}(C \imp A)$ is represented in Figure \ref{fig:ex-sym}.
    Since $\phi_C$ is transitive, we have $\phi_C(t_3[C], t_j[C]) = \ctt{true}$ for each tuple $t_j$ of $r$.
    Moreover, $\phi_A(t_3[A], t_6[A]) = \ctt{false}$ since $(t_3, t_6)$ is a counterexample to $C \imp A$.
    Therefore, the transitivity of $\phi_A$ implies either $\phi_A(t_3[A], t_4[A]) = \ctt{false}$ or $\phi_A(t_4[A], t_6[A]) = \ctt{false}$.
    Hence, at least one of $(t_3, t_4)$ and $(t_4, t_6)$ must be a counterexample to $C \imp A$ too.
    In the example, this is $(t_3, t_4)$.
    \begin{figure}[ht!]
        \centering 
        \includegraphics[scale=1.0, page=3]{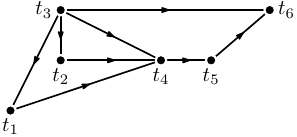}
        \caption{The conflict-graph $\CG_{\Phi}(r, C \imp A)$ of Example \ref{ex:ex-sym}.}
        \label{fig:ex-sym}
    \end{figure}
\end{example}
Nevertheless, if transitivity constrains the complexity of the graph,
dropping symmetry still allows new kinds of graph structures.
Indeed, in the presence of symmetry, a conflict-graph cannot contain induced 
paths with more than $3$ elements because of transitivity.
However, such paths may exist when symmetry is removed.
\begin{example}
In the previous example, the tuples $t_2, t_4, t_5, t_6$ form an induced $P_4$ of the 
underlying undirected graph of $\CG_{\Phi}(r, C \imp A)$, even though $\phi_A$ and 
$\phi_C$ enjoy transitivity.
\end{example}
Therefore, we are left with the following intriguing question: can the loss of symmetry   
be used to break transitivity, and offer conflict-graphs a structure sufficiently complex 
to make \csmc{EVPP} intractable?
The next theorem answers this question affirmatively.

\begin{theorem} \label{thm:vpp-sym}
The problem \csmc{EVPP} is \NP-complete even when the predicates used on each 
attribute are transitive (\ctt{tra}), reflexive (\ctt{ref}), and antisymmetric (\ctt{asym}).
\end{theorem}

The proof is given in Appendix \ref{app:vpp-sym}.
It is a reduction from MIS in $2$-subdivision graphs \cite{poljak1974note}.

Theorem \ref{thm:vpp-ref}, Theorem \ref{thm:vpp-tra} and Theorem \ref{thm:vpp-sym} 
characterize the complexity of \csmc{EVPP} for each combination of predicates properties.
In the next section, we discuss the granularity of these, and we use them as a framework 
to compare the complexity of \csmc{EVPP} for some known extensions of functional 
dependencies.

\section{Discussions}
\label{sec:discussions}

Replacing equality with various predicates to extend the semantics of classical functional 
dependencies is frequent \cite{caruccio2015relaxed,song10tree}.
Our approach offers to compare these extensions on \csmc{EVPP} within a unifying 
framework based on the properties of the predicates they use.
We can summarize our results with the hierarchy of classes of predicates given 
in Figure \ref{fig:summary}.
\begin{figure}[ht!]
\centering 
\includegraphics[scale=1.0, page=1]{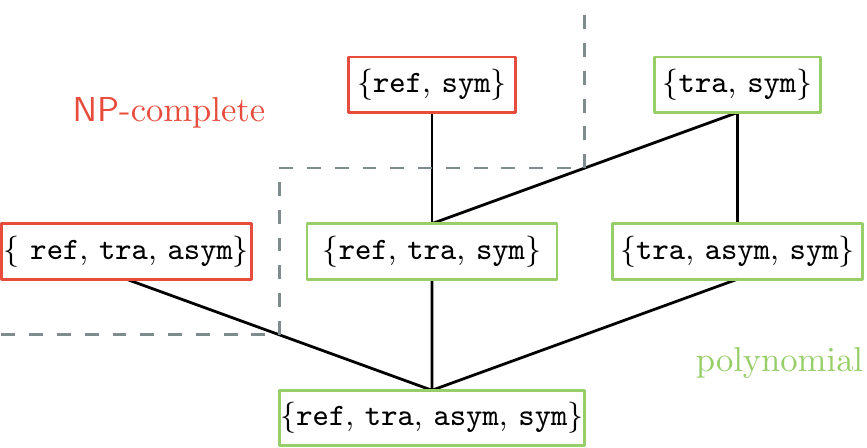}%
\caption{Complexity of \csmc{EVPP} with respect to the properties of predicates.}
\label{fig:summary}
\end{figure}

Regarding the computation of the $g_3$-error, most existing works have focused on
similarity/distance predicates.
First, the $g_3$-error can be computed in polynomial time for classical functional 
dependencies \cite{huhtala1998efficient}.
Then, Song et al. \cite{song2013comparable} show that \csmc{EVPP} is \NP-complete 
for a broad range of extensions of FDs which happen to be reflexive (\ctt{ref}) 
and symmetric (\ctt{sym}) predicates, which coincides with Theorem \ref{thm:vpp-tra}.
However, they do not study predicate properties as we do in this paper.
More precisely, they identify the hardness of \csmc{EVPP} for differential 
\cite{song2010data}, matching \cite{fan2008dependencies}, metric \cite{koudas2009metric}, 
neighborhood \cite{bassee2001neighborhood}, and comparable dependencies 
\cite{song2013comparable}.
For some of these dependencies, predicates may be defined over sets of 
attributes.
Using one predicate per attribute and taking their conjunction is a particular case of 
predicate on attribute sets. 

Some extensions of FDs use partial orders as predicates.
This is the case of ordered dependencies \cite{dong1982applying,ginsburg1983order}, 
 ordered FDs \cite{ng2001extension}, and 
also of some sequential dependencies \cite{golab2009sequential} and denial constraints 
\cite{bertossi2005complexity} for instance.
To our knowledge, the role of symmetry in \csmc{EVPP} has received little attention.
%Some particular cases have been tackled \cite{}, and some
%algorithms have been recently devised \cite{karegar2021efficient}.
%In \cite{dong1982applying}, the author study particular case of ordered dependencies.
For sequential dependencies \cite{golab2009sequential}, a measure different than the 
$g_3$-error have been used.
The predicates of Theorem \ref{thm:vpp-sym} are reflexive, transitive and antisymmetric.
Hence they are partial orders.
Consequently, the FDs in this context are \emph{ordered 
functional dependencies} as defined by Ng \cite{ng2001extension}.
We obtain the following corollary:
\begin{corollary} \label{cor:ods}
\csmc{EVPP} is \NP-complete for ordered functional dependencies.
\end{corollary}
Ordered functional dependencies are a restricted case of ordered dependencies 
\cite{ginsburg1983order}, 
sequential dependencies \cite{golab2009sequential}, and denial constraints 
\cite{bertossi2005complexity} (see \cite{song10tree}).
The hardness of computing the $g_3$-error for these dependencies follows from Corollary 
\ref{cor:ods}.

The hierarchy depicts quite accurately the current knowledge about 
\csmc{EVPP} and the delimitation between tractable and intractable cases.
However, this analysis may require further refinements.
Indeed, there may be particular types of FDs with predicates where 
\csmc{EVPP} is tractable in polynomial time, even though their predicates belong to a 
class for which the problem is \NP-complete.
For instance, assume that each attribute $A$ in $\R$ is equipped with a \emph{total} 
order $\phi_A$.
We show in Proposition \ref{prop:total-order} and Corollary \ref{cor:total-order} that in 
this case, \csmc{EVPP} can be solved in polynomial time, even though the predicates are 
reflexive, transitive and antisymmetric.

\begin{proposition} \label{prop:total-order}
Let $(\R, \Phi)$ be a relation scheme with predicates.
Then, \csmc{EVPP} can be solved in polynomial time for a given FD $X \imp A$ if $\phi_B$ 
is transitive for each $B \in X$ and $\phi_A$ is a total order.
\end{proposition}

\begin{proof}
Let $(\R, \Phi)$ be a relation scheme with predicates and $X \imp A$ a functional 
dependency.
Assume that $\phi_B$ is transitive for each $B \in X$ and that $\phi_A$ is a total order.
Let $r$ be a relation over $(\R, \Phi)$.
Let $\G = (r, \E)$ be the undirected graph underlying $\CG_{\Phi}(r, X \imp A)$, that is, 
$(t_i, t_j) \in \E$ if and only if $(t_i, t_j)$ or $(t_j, t_i)$ is an edge 
of $\CG_{\Phi}(r, X \imp A)$.

We show that $\G$ is a comparability graph.
To do so, we associate the following predicate $\leq$ to $\CG_{\Phi}(r, X \imp A)$: for 
each pair $t_i, t_j$ of tuples of $r$, $t_i \leq t_i$ and $t_i \leq t_j$ if $(t_i, t_j)$ 
is a counterexample to $X \imp A$.
We show that $\leq$ is a partial order:
\begin{itemize}
\item \emph{reflexivity}. It follows by definition.

\item \emph{antisymmetry}. We use contrapositive.
Let $t_i, t_j$ be two distinct tuples of $r$ and assume that $(t_i, t_j)$ belongs to
$\CG_{\Phi}(r, X \imp A)$.
We need to prove that $(t_j, t_i)$ does not belong to $\CG_{\Phi}(r, X \imp A)$, \ie it 
is not a counterexample to $X \imp A$.
First, $(t_i, t_j) \in \CG_{\Phi}(r, X \imp A)$ implies that $\phi_A(t_i[A], t_j[A]) = 
\ctt{false}$.
Then, since $\phi_A$ is a total order, $\phi_A(t_j[A], t_i[A]) = \ctt{true}$.
Consequently, $(t_j, t_i)$ cannot belong to $\CG_{\Phi}(r, X \imp A)$ and $\leq$ is 
antisymmetric.

\item \emph{transitivity}. Let $t_i, t_j, t_k$ be tuples of $r$ such that $(t_i, t_j)$ 
and $(t_j, t_k)$ are in $\CG_{\Phi}(r, X \imp A)$.
Applying transitivity, we have that $\Land_{B \in X} 
\phi_B(t_i[B], t_k[B]) = \ctt{true}$.
We show that $\phi_A(t_i[A], \allowbreak t_k[A]) = \ctt{false}$.
Since $(t_i, t_j)$ is a counterexample to $X \imp A$, we have $\phi_A(t_i[A], t_j[A]) = 
\ctt{false}$.
As $\phi_A$ is a total order, we deduce that $\phi_A(t_j[A], t_i[A]) = \ctt{true}$.
Similarly, we obtain $\phi_A(t_k[A], \allowbreak t_j[A]) = \ctt{true}$.
As $\phi_A$ is transitive, we derive $\phi_A(t_k[A], \allowbreak t_i[A]) = \ctt{true}$.
Now assume for contradiction that $\phi_A(t_i[A], t_k[A]) = \ctt{true}$.
Since, $\phi_A(t_k[A], t_j[A]) = \ctt{true}$, we derive $\phi_A(t_i[A], t_j[A]) = 
\ctt{true}$ by transitivity of $\phi_A$, a contradiction.
Therefore, $\phi_A(t_i[A], \allowbreak t_k[A]) = \ctt{false}$.
Using the fact that $\Land_{B \in X} \allowbreak \phi_B(t_i[B], t_k[B]) = \ctt{true}$, we 
conclude that $(t_i, t_k)$ is also a counterexample to $X \imp A$.
The transitivity of $\leq$ follows.
\end{itemize}
Consequently, $\leq$ is a partial order and $\G$ is indeed a comparability graph.
Since \csmc{MIS} can be solved in polynomial time for comparability graphs 
\cite{golumbic2004algorithmic}, the result follows.
\end{proof}

We can deduce the following corollary on total orders, that can be used for ordered 
dependencies.

\begin{corollary} \label{cor:total-order}
Let $(\R, \Phi)$ be a relation scheme with predicates.
Then, \csmc{EVPP} can be solved in polymomial time if each predicate in $\Phi$ is a total order.
\end{corollary}

In particular, Golab et al. \cite{golab2009sequential} proposed a polynomial-time 
algorithm for a variant of $g_3$ applied to a restricted type of sequential dependencies 
using total orders on each attribute.

\section{Conclusion and future work}
\label{sec:conclusion}

In this work, we have studied the complexity of computing the $g_3$-error when equality 
is replaced by more general predicates.
We studied four common properties of binary predicates: reflexivity, symmetry, 
transitivity, and antisymmetry.
We have shown that when symmetry and transitivity are taken together, the $g_3$-error can 
be computed in polynomial time.
Transitivity strongly impacts the structure of the conflict-graph of the 
counterexamples to a functional dependency in a relation.
Thus, it comes as no surprise that dropping transitivity makes the $g_3$-error hard to 
compute.
More surprisingly, removing symmetry instead of transitivity leads to the same conclusion.
This is because deleting symmetry makes the conflict-graph directed.
In this case, the orientation of the edges weakens the impact of transitivity, thus 
allowing the conflict-graph to be complex enough to make the $g_3$-error computation 
problem intractable. 

We believe our approach sheds new light on the problem of computing the $g_3$-error, and 
that it is suitable for estimating the complexity of this problem when defining new types 
of FDs, by looking at the properties of predicates used to compare values.

We highlight now some research directions for future works.
In a recent paper \cite{livshits2020computing}, Livshits et al. study the problem of 
computing optimal repairs in a relation with respect to a set of functional dependencies. 
A repair is a collection of tuples which does not violate a prescribed set of FDs.
It is optimal if it is of maximal size among all possible repairs.
Henceforth, there is a strong connection between the problem of computing repairs and 
computing the $g_3$-error with respect to a collection of FDs.
In their work, the authors give a dichotomy between tractable and intractable cases based 
on the structure of FDs.
In particular, they use previous results from Gribkoff et al. \cite{gribkoff2014most} to  
show that the problem is already $\NP$-complete for $2$ FDs in general.
In the case where computing an optimal repair can be done in polynomial time, it would be 
interesting to use our approach and relax equality with predicates in order to study the 
tractability of computing the $g_3$-error on a collection of FDs with relaxed equality.

From a practical point of view, the exact computation of the $g_3$-error is 
extremely expensive in large datasets.
Recent works \cite{caruccio2016discovery,faure2022assessing} have 
proposed to use approximation algorithms to compute the $g_3$-error both for equality 
and predicates. 
It could be of interest to identify properties or classes of predicates where more 
efficient algorithms can be adopted.
It is also possible to extend the existing algorithms calculating the  
classical $g_3$-error (see \eg \cite{huhtala1999tane}).
They use the projection to identify equivalence classes among values of $A$ and $X$.
However, when dropping transitivity (for instance in similarity predicates), separating 
the values of a relation into \textit{``similar classes''} requires to devise a new 
projection operation, a seemingly tough but fascinating problem to investigate.

\paragraph{Acknowledgment.} We thank the reviewers for their valuable feedback..
We also thank the Datavalor initiative of Insavalor (subsidiary of INSA Lyon) for funding part of this work.

\bibliographystyle{alpha}
\bibliography{biblio}
\newpage
\appendix

\section{Proof of Theorem \ref{thm:vpp-tra}}
\label{app:vpp-tra}

\begin{theorem*}[\ref{thm:vpp-tra}]
    The problem \csmc{EVPP} is \NP-complete even when the predicates used on each attributes are symmetric 
    (\ctt{sym}) and reflexive (\ctt{ref}).
\end{theorem*}

\begin{proof}
We first show that \csmc{EVPP} belongs to \NP.
Let $(\R, \Phi)$ be a relation scheme with predicates, $r$ a relation over $(\R, \Phi)$, $X \imp 
A$ a functional dependency over $(\R, \Phi)$ and $k \in \cb{R}$.
We have that $g_3^{\Phi}(X \imp A, r) \leq k$ if and only if there exists a subrelation 
$s$ in $r$ satisfying $s \simmodels X \imp A$ and $1 - \frac{\card{s}}{\card{r}} \leq 
k$, or $\card{s} \geq (1 - k) \times \card{r}$ equivalently.
Therefore, a certificate for \csmc{EVPP} is a subrelation $s$ containing at least $(1 - 
k) \times \card{r}$ tuples and satisfying $X \imp A$ (with respect to $\Phi$).
Since predicates can be computed in polynomial time by assumption, it takes polynomial 
time to check that $s \simmodels X \imp A$.
Thus, \csmc{EVPP} belongs to \NP.

To show \NP-completeness, it is convenient to use a reduction from Maximum Clique (\csmc{MC}) rather than \csmc{MIS}, even though the problems are polynomially equivalent:
\Problem{Maximum Clique (\csmc{MC})}
{A graph $\G = (\V, \E)$, $k \in \cb{N}$.}
{\ctt{yes} if $\G$ has a clique with at least $k$ vertices.} 
Let $\G = (V, \E)$ be a graph with $V = \{u_1, \dots, u_n\}$ for some $n \in \cb{N}$, and $\E = \{e_1, \dots, e_m\}$ for some 
$m \in \cb{N}$.
Let $k$ be an integer such that $k \leq \card{V}$.
We construct an instance of \csmc{EVPP}.
We begin with a relation scheme with predicates $(\R, \Phi)$ where $\R = \{B_1, \dots, 
B_m, A\}$, $\Phi = \{\phi_1, \dots, \phi_m, \phi_A\}$ , and:
\begin{itemize}
\item for each $1 \leq i \leq m$, $\dom(B_i) = \{0, 1, 2\}$ and $\phi_i$ is defined as 
follows:
\[ 
\phi_i(x, y) = \begin{cases}
    \ctt{true} & \text{if } x = y \text{ or } x + y < 3 \\
    \ctt{false} & \text{otherwise.}
\end{cases}
\]
Observe that $\phi_i$ is reflexive and symmetric.
\item $\dom(A) = \{1, \dots, n\}$, and the predicate $\phi_A$ for $A$ is defined by 
$\phi_A(x, y) = \ctt{true}$ if and only if $x = y$.
Thus, $\phi_A$ is reflexive and symmetric.
\end{itemize}
Observe that the predicates can be computed in polynomial time in 
the size of their input.
Now, we build a relation $r = \{t_1, \dots, t_n\}$ (one tuple per vertex in $\G$) over 
$(\R, \Phi)$.
For each $1 \leq i \leq n$, we put $t_i[A] = i$ and for each $1 \leq j \leq m$:
\[ 
    t_i[B_j] = 
    \begin{cases}
        0 & \text{if } u_i \notin e_j \\
        1 & \text{if } e_j = (u_i, u_{\ell}) \text{ and } i < \ell \\
        2 & \text{if } e_j = (u_{\ell}, u_i) \text{ and } \ell < i \\
    \end{cases}
\]
Finally, let $k' = 1 - \frac{k}{n}$, and consider the functional dependency $X \imp A$ 
where $X = \{B_1, \dots, B_m\}$.
We obtain an instance of \csmc{EVPP} which can be constructed in polynomial time 
in the size of $\G$.
The reduction is illustrated on an example in Figure \ref{fig:vpp-tra}.
\begin{figure}[ht!]
    \centering 
    \includegraphics[scale=1]{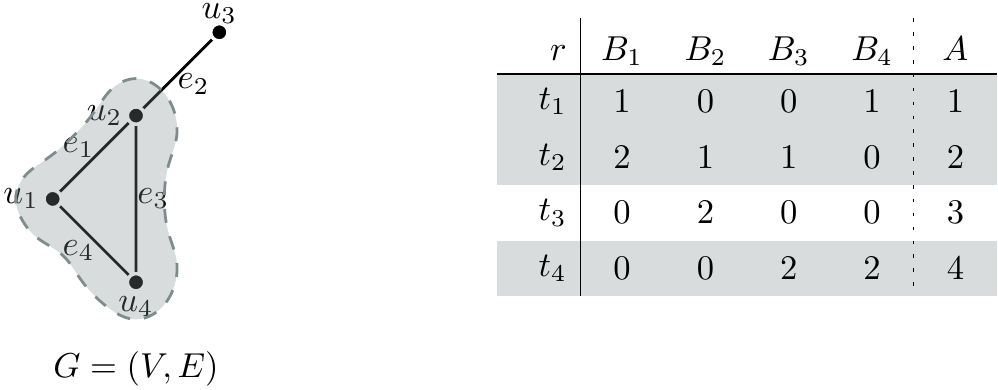}%
    \caption{Reduction of Theorem \ref{thm:vpp-tra}. In grey, a clique and its 
    associated subrelation satisfying $X \imp A$.}
    \label{fig:vpp-tra}
\end{figure}

To conclude the proof, we have to prove that $\G$ contains a clique $K$ such that 
$\card{K} \geq k$ if and only if $g_3^{\Phi}(X 
\imp A, r) \leq k'$. 
To do so, we show that for every distinct tuples $t_i, t_j$ of 
$r$, $\bigwedge_{1 \leq \ell \leq m} \phi_{\ell}(t_i[B_{\ell}], \allowbreak t_j[B_{\ell}]) = \ctt{true}$ if 
and only 
if $(u_i, u_j)$ is not an edge of $\G$.
% This implies that a subset $K$ of $\V$ is a clique in 
% $\G$ if and only if the corresponding set of tuples $s(K)$ is a subrelation of $r$ 
% which satisfies $X \imp A$.

We begin with the only if part.
Hence, assume that for each $1 \leq \ell \leq m$, we have $\phi_{\ell}(t_i[B_{\ell}], t_j[B_{\ell}]) = 
\ctt{true}$.
By definition of $\phi_{\ell}$, we have two cases:
\begin{itemize}
    \item $t_i[B_{\ell}] = t_j[B_{\ell}]$.
    By construction of $r$, it follows that $t_i[B_{\ell}] = 0$.
    Hence, neither $u_i$ nor $u_j$ belongs to $e_{\ell}$.
    
    \item $t_i[B_{\ell}] + t_j[B_{\ell}] < 3$.
    It follows that either $t_i[B_{\ell}] = 0$ or $t_j[B_{\ell}] = 0$.
    Without loss of generality, assume that $t_i[B_{\ell}] = 0$.
    Then, again by construction of $r$, we deduce that $u_i \notin e_{\ell}$. 
\end{itemize}
Thus, $(u_i, u_j)$ is not an edge of $\G$.

We move to the if part.
We use contrapositive.
Hence, assume there exists some $B_{\ell}$, $1 \leq \ell \leq m$, such that $\phi_{\ell}(t_i[B_{\ell}],  
t_j[B_{\ell}]) = \ctt{false}$.
By definition of $\phi_{\ell}$, we deduce that $t_i[B_{\ell}] \neq t_j[B_{\ell}]$ and $t_i[B_{\ell}] + 
t_j[B_{\ell}] \geq 3$.
Without loss of generality, we obtain $t_i[B_{\ell}] = 1$ and $t_j[B_{\ell}] = 2$.
Therefore, by construction of $r$, $u_i \in e_{\ell}$  and $u_j \in e_{\ell}$ must hold.
As $t_i[B_{\ell}] \neq t_j[B_{\ell}]$, we deduce that $(u_i, u_j) = e_{\ell}$, concluding this part of the proof.

Consequently, a subset $K$ of $\V$ is a clique in 
$\G$ if and only if the corresponding set of tuples $s(K)$ is a subrelation of $r$ 
which satisfies $X \imp A$.
Therefore, $\G$ contains a clique $K$ such that 
$\card{K} \geq k$ if and only if $g_3^{\Phi}(X \imp A, r) \leq k'$ holds, which concludes 
the proof. \end{proof}

\section{Proof of Theorem \ref{thm:vpp-sym}}
\label{app:vpp-sym}

\begin{theorem*}[\ref{thm:vpp-sym}]
The problem \csmc{EVPP} is \NP-complete even when the predicates used on each 
attribute are transitive (\ctt{tra}), reflexive (\ctt{ref}), and antisymmetric (\ctt{asym}).
\end{theorem*}

\begin{proof}
The fact that \csmc{EVPP} belongs to $\NP$ has been shown in Theorem \ref{thm:vpp-tra}.

To show \NP-completeness, we use a reduction from \csmc{MIS} in $2$-subdivision graphs, in which \csmc{MIS} remains \NP-complete \cite{poljak1974note}.
Let $\G = (V, \E)$ be an (undirected) graph where $V = \{u_1, \dots, u_n\}$ and $\E = 
\{e_1, \dots, e_m\}$.
Without loss of generality, we assume that $\G$ is loopless and that each vertex belongs 
to at least one edge.
Let $\V_2 = \V \cup \{v_k^i \mid 1 \leq k \leq m, 1 \leq i 
\leq n \text{ and } u_i \in e_k\}$ be a new set of vertices.
We construct a set $\E_2$ of edges.
It is obtained from $\E$ by replacing each edge $e_k = (u_i, 
u_j)$ by a path made of three edges $\{(u_i, v_k^i), (v_k^i, v_k^j), (v_k^j, u_j)\}$.
The graph $\G_2 = (\V_2, \E_2)$ is the $2$-subdivision of $\G$.
%Now let $\G_2 = (V_2, \E_2)$ with $V_2 = V \cup \{v_k^i \mid 1 \leq k \leq m, 1 \leq i 
%\leq n \text{ and } u_i \in e_k\}$.
%Then, $\E_2$, is obtained from $\E$ by replacing each edge $e_k = (u_i, u_j)$ by a path 
%made of three edges $\{(u_i, v_k^i), (v_k^i, v_k^j), (v_k^j, u_j)\}$.
%The graph $\G_2$ is the $2$-subdivision of $\G$.
Every $2$-subdivision graph is the $2$-subdivision of some graph.
The graph $\G_2$ can be built in polynomial time in the size of $\G$.

Now we construct an instance of \csmc{EVPP}.
Let $\{a_1, \dots, a_n\}$ be a set of characters.
We build a relation scheme with predicates $(\R, \Phi)$ where $\R = \{B, A\}$, 
$\Phi = \{\phi_B, \phi_A\}$, and:
\begin{itemize}
\item $\dom(B)$ is the set of pairs of symbols associated to $\{a_1, \dots, 
a_n\} \times \{a_1, \dots, a_n\}$.
We add a predicate $\phi_B$ as follows:
\[
\phi_B(x, y) = \begin{cases}
\ctt{true} & \text{if } x = y \\
\ctt{true} & \text{if } x \neq y \text{ and } x[1] = x[2] \text{ and } x[1] \in \{y[1], 
y[2]\} \\
\ctt{false} & \text{otherwise.}
\end{cases}
\]
The predicate is \emph{reflexive} by definition.
We prove that it is \emph{transitive}.
Let $x, y, z \in \dom(B)$ and assume that $\phi_B(x, y) = \phi_B(y, z) = \ctt{true}$.
If $x = y = z$, we readily have $\phi_B(x, z) = \ctt{true}$.
Since $x \neq z$ implies $x \neq y$ or $y \neq z$, it is sufficient to show that $\phi(x, 
z) = \ctt{true}$ in these two cases.
Assume first that $x \neq y$.
Then $\phi_B(x, y) = \ctt{true}$ if and only if $x = a_i a_i$ and $y \in \{a_ia_j, 
a_ja_i\}$ for $1 \leq i, j \leq n$, $i \neq j$.
It follows that $\phi_B(y, z)$ holds if and only if $z = y$. 
Thus, $\phi_B(x, z) = \phi_B(x, y) = \ctt{true}$ is valid.
Let us assume now that $y \neq z$.
Then, $\phi_B(y, z) = \ctt{true}$ implies that $y = a_i a_i$ for some $1 \leq i \leq n$, 
by definition of $\phi_B$.
Therefore, $\phi_B(x, y) = \ctt{true}$ entails $x = y$.
We deduce $\phi_B(x, z) = \ctt{true}$. 
Consequently, $\phi_B$ is transitive.
%For every $x \in \dom(B)$ such that $x = a_i a_j$ with distinct $i$ and $j$, $\phi_B(x, 
%y) = \ctt{true}$ only if $x = y$.
%Therefore, 
%To see that it is transitive, observe that $\phi_B(x, y) = \phi_B(y, z) = \ctt{true}$ 
%implies $x = y$ or $y = z$.
%Therefore, the transitivity of $\phi_B$ follows from its reflexivity.
At last, assume that $\phi_B(x, y) = \ctt{true}$ with $x \neq y$.
Hence, $y[1] \neq y[2]$ and $\phi_B(y, x)$ cannot be true.
Therefore, $\phi_B(x, y) = \phi_B(y, x) = \ctt{true}$ entails $x = y$.
Thus, $\phi_B$ is also \emph{antisymmetric}.

\item $\dom(A) = \{1, \dots, n\}$ and $\phi_A(x, y) = \ctt{true}$ if and only if $x = y$.
In other words, $\phi_A$ is the usual equality.
Hence, it enjoys both reflexivity, transitivity and antisymmetry.
\end{itemize}
Observe that all predicates can be computed in polynomial time in the size of 
their input.
%Moreover, they can be stored as tables whose size are polynomial in the size of $\G_2$.
Now we construct a relation $r = \{t_1, \dots, t_n\} \cup \{t_k^i \mid 1 \leq k \leq m, 1 
\leq i \leq n, v_k^i \in V_2\}$ (one tuple per vertex in $\G_2$) over $(\R, \Phi)$:
\begin{itemize}
\item for each $1 \leq i \leq n$, $t_i[B] = a_i a_i$ and $t_i[A] 
= i$,
\item for each $1 \leq k \leq m$ and each $1 \leq i \leq n$ such that $v_k^i \in V_2$, 
let $e_k = (u_i, u_j)$, $1 \leq j \leq n$, be the corresponding edge of $\G$.
Then, we put $t_k^i[B] = a_i a_j$ if $i < j$ and $a_j a_i$ otherwise.
As for $A$, we define $t_k^i[A] = j$.
\end{itemize}
Finally, we consider the functional dependency $B \imp A$.
The whole reduction can be computed in polynomial time in the size of $\G$.
It is illustrated on an example in Figure \ref{fig:vpp-sym}.
Intuitively, $\phi_B$ guarantees that two tuples representing adjacent vertices of $\G_2$ will agree on $B$ in $(\R, \Phi)$.
    However, the transitivity of $\phi_B$ will produce pairs of tuples which agree on $B$ even though they are not adjacent in $\G_2$.
    More precisely, $\phi_B$ returns $\ctt{true}$ in two cases:
    \begin{itemize}
        \item  when it compares $t_i$ to $t_k^i$ and $t_k^j$ for each edge $e_k$ of $\G$ to which 
        $u_i$ belongs, and
        \item when it compares $t_k^i$ to $t_k^j$ for each edge $e_k$ of $\G$.
    \end{itemize}
    The role of $\Phi_A$ is then to assert that non-adjacent tuples cannot produce counterexamples.

\begin{figure}[ht!]
\centering
\includegraphics[scale=1.0, page=3]{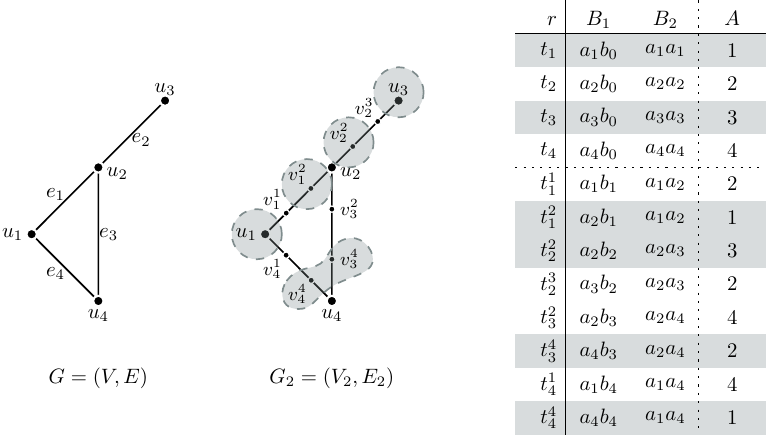}%
\caption{Reduction of Theorem \ref{thm:vpp-sym}. In grey, an independent set and its 
associated subrelation satisfying $B \imp A$.}
\label{fig:vpp-sym}
\end{figure}

We begin with the if part.
Consider two (distinct) vertices of $V_2$ that are connected in $\G_2$.
Because $\G_2$ is the $2$-subdivision of $\G$, we have the following cases:
\begin{itemize}
\item $u_i, v_k^i$ for some $1 \leq i \leq n$ and $1 \leq k \leq m$.
For $B$, we have $t_i[B] = a_i a_i$ and $t_k^i[B] = a_i a_j$ (or $a_j a_i$) for 
some $1 \leq j \leq n$.
Therefore, $\phi_B(t_i[B], t_k^i[B]) = \ctt{true}$ holds.
However, $t_i[A] \neq t_k^i[A]$ also by definition of $r$.
Thus, $\{t_i, t_k^i\} \not\simmodels B \imp A$.
    
\item $v_k^i, v_k^j$ for some $1 \leq i < j \leq n$ (without loss of generality) and $1 
\leq k \leq m$.
Then, $t_k^i[B] = t_k^j[B]$, $t_k^i[A] = j$, and $t_k^j[A] = i$.
It follows that $\{t_k^j, t_k^i\} \not\simmodels B \imp A$, by definition of $\phi_B$ and 
$\phi_A$.
\end{itemize}
Thus, if two vertices are connected in $\G_2$, the corresponding tuples in $r$ do not 
satisfy the 
functional dependency $B \imp A$, concluding this part of the proof.

We show the only if part using contrapositive.
Consider two distinct vertices of $V_2$ that are not connected in $\G_2$.
We have four cases:
\begin{itemize}
\item $u_i, u_j$ for some $1 \leq i < j \leq n$.
By definition of $r$, we have $t_i[B]= a_i a_i$ and $t_j[B] = a_j a_j$.
Thus, $\phi_B(t_i[B], t_j[B]) = \phi_B(t_j[B], t_i[B]) \allowbreak = \ctt{false}$, and 
$\{t_i, t_j\} \simmodels B \imp A$ holds.

\item $v_k^i, v_{\l}^{j}$ for some $1 \leq k, \l \leq m$ and $1 \leq i, j \leq n$.
Then, $t_k^i[B][1] \neq t_k^i[B][2]$ and $t_{\l}^{j}[B][1] \neq t_{\l}^{j}[B][2]$. 
According to $\G_2$, $v_k^i$ and $v_{\l}^{j}$ are not connected if and only if $k \neq 
\l$.
Consequently, $t_k^i[B] \neq t_{\l}^j[B]$ by definition of $r$.
Hence, $\phi_B(t_k^i[B], t_{\l}^{j}[B]) = \phi_B(t_{\l}^{j}[B], t_k^i[B]) = 
\ctt{false}$.
We deduce that $\{t_k^i, t_{\l}^{j}\} \simmodels B \imp A$.

\item $u_i, v_k^j$ for some $1 \leq i, j \leq n$, $1 \leq k \leq m$ and $u_i \notin e_k$ 
in $\G$.
Then, $t_i[B] = a_i a_i$ and since $u_i \notin e_k$, we have $t_k^j[B] = a_j a_{\l}$ 
(or $a_{\l} a_j$) for some $1 \leq \l \leq n$ and
$i \neq j, \l$.
By definition of $\phi_B$, we deduce that $\phi_B(t_i[B], t_k^j[B]) = 
\phi_B(t_k^j[B], t_i[B]) = \ctt{false}$ must hold.
Therefore, $\{t_i, t_j^k\} \simmodels B \imp A$ is true too.

\item $u_i, v_k^j$ for some $1 \leq i, j \leq n$, $1 \leq k \leq m$ and $u_i \in e_k$ in 
$\G$.
Then, necessarily $i \neq j$ by construction of $\G_2$.
Consequently, we must have $t_i[A] = t_j^k[A] = i$ by definition of $r$.
Therefore, $\phi_A(t_i[A], t_k^j[A]) = \ctt{true}$ and $\{t_i, t_k^j\} \simmodels B 
\imp 
A$ holds.
\end{itemize}

Thus, whenever two vertices of $\G_2$ are disconnected, the corresponding set of tuples 
of $r$ satisfies $B \imp A$.
This concludes the proof of the equivalence.

Consequently, $\G_2$ has an independent set of size $k$ if and only if there exists a 
subrelation $s$ of $r$ of size $k$ which satisfies $B \imp A$, 
concluding the proof.
\end{proof}

\end{document}